\documentclass[12pt]{article}

\usepackage{amsmath,amssymb,amsthm}
\usepackage{mathtools}
\usepackage{hyperref}
\usepackage{natbib}
\usepackage{enumitem}
\usepackage{booktabs}
\usepackage{xcolor}
\usepackage{graphicx}

\usepackage{newpxtext}  
\usepackage{newpxmath}
\usepackage[left=2cm,right=2cm,top=2cm,bottom=2cm]{geometry}

\usepackage{tabularx}

\hypersetup{colorlinks=true, linkcolor=blue!50!black, citecolor=blue!50!black, urlcolor=blue!50!black}

\newtheorem{definition}{Definition}
\newtheorem{proposition}{Proposition}
\newtheorem{remark}{Remark}

\title{Interventional Score Geometry for Causal Inference:\\
\large A Dictionary for Structural Identification}
\author{
Mojtaba Eslami\thanks{
Post-doctoral Researcher, University of Calgary.
Email: \texttt{mojtaba.eslami@alumni.ucalgary.ca}}}

\date{\small{May 21, 2026} \\ \vspace{0.3cm} \small Version 3.2}

\begin{document}
\maketitle

\begin{abstract}
Let $p(x)$ denote the joint density of a vector of observed variables $X$, and let $\psi(x) = \nabla_x \log p(x)$ be its score field, the gradient of the log-density. The observational geometry built from $p(x)$ and $\psi(x)$ alone cannot identify causal direction, because it is a function of $p$ alone: any two structural models that share $p$ share this geometry exactly, in the same way that two rivers merging downstream can no longer be told apart by sampling the water below the merge. I develop the interventional analogue carefully. A hard intervention $\operatorname{do}(X_k=\xi)$, which fixes coordinate $X_k$ at the value $\xi$ by policy or design, does not merely reweight the joint law. It collapses the law onto the submanifold $\{x_k=\xi\}$, much as clamping one string of an instrument does not just change its pitch but removes that string entirely from the space of things that can vibrate. Its score is therefore not a gradient in the original $d$ coordinates. I define it instead on the $d-1$ coordinates left free to vary, where it is an ordinary, non-singular object. Causal influence $X_k \rightsquigarrow X_j$ is defined distributionally, as variation of the interventional \emph{marginal} law of $X_j$ with $\xi$, and a matching derivative of the marginal interventional score gives a locally checkable sufficient condition for it. I also show that projecting the observational score onto a set of admissible intervention directions (the coordinates that a given design or policy can, in principle, manipulate) does not recover the causal response in general: two structural models can share an observational score and an admissible set while responding quite differently to intervention. I replace that projection with an interventional response field that must instead be supplied from structural information. A causal metric is then built as an ordinary Fisher information metric on a family of interventions sharing a common target, which avoids the ill-posed cross-target comparisons a naive construction runs into. I use this vocabulary to build a corrected geometric dictionary for randomized trials, instrumental variables, and conditional-independence designs, stating precisely what each does and does not imply, and I illustrate the observational and interventional distinction with a fully worked bivariate Gaussian example in which two models share an observational score but differ sharply in their interventional score derivative. The framework is deliberately narrow. It supplies notation for relating designs, interventions, and score fields, and it does not manufacture identification power beyond what the underlying structural or design assumptions already provide. It is also compatible with Pearl's Ladder of Causation: observational score geometry belongs to the associational rung, intervention-indexed score fields and response operators belong to the interventional rung, and a geometry of unit-level counterfactuals is left for future work.
\end{abstract}

\section{Introduction}

Structural causal analysis distinguishes between the joint distribution of observables and the causal mechanism that generates it. Distinct structural models can generate the same reduced-form distribution while disagreeing about causal direction; this classical non-identification result motivates instrumental variables, exclusion restrictions, and randomized designs \citep{haavelmo1943,pearl2009,imbens2015}.

Information geometry equips statistical models with a Riemannian structure derived from the Fisher information metric \citep{amari2016}, and recent work studies causal and predictive relationships through the geometry of manifolds built from observational data \citep{surasinghe2020}. These are, by construction, functions of $p(x)$ alone, and therefore face the same identification limit as $p(x)$ itself.

This paper's organizing idea is that causal asymmetry has to come from outside $p(x)$, the joint density of the observed variables: from an explicit intervention. A distribution, however carefully you stare at it, tells you what tends to go with what; it does not tell you what would happen if you reached in and moved something. That has to be supplied, and once it is, the interesting question is how to write it down without smuggling in more than was actually given. Section~\ref{sec:interventional} builds the resulting interventional score correctly, on the coordinates left non-degenerate by a hard intervention, and works through a fully explicit example. Section~\ref{sec:response} shows why an admissible set of intervention directions cannot, by itself, convert an observational score into a causal one, and introduces the interventional response field that must replace it. Section~\ref{sec:metric} builds a causal metric as ordinary Fisher information on a family of interventions sharing a common target, so that the underlying densities live on a common space. Section~\ref{sec:designs} restates randomization, instrumental variables, and conditional-independence designs in this vocabulary, stating precisely what follows from each assumption and what does not.

\subsection{A geometric motivation}
\label{subsec:motivation}

It helps to have a picture in mind before the definitions start. The intuition I are borrowing, loosely, is the one general relativity uses for gravity: gravity is not an extra force bolted onto an otherwise fixed background, but something read off from how the geometry of spacetime governs motion. What is locally felt as gravitational pull reflects a larger structure, not a property sitting inside one object. I want to treat causal influence the same way: not as a label attached to a single observational distribution, but as something read off from how a whole family of distributions moves as an intervention parameter changes. The relevant object is therefore never $(M,g,\psi)$ sitting still; it is the indexed family $\xi\mapsto p_{-k}^{(k,\xi)}$, and how that family deforms as $\xi$ varies.

The analogy is worth exactly as much as it is limited, and it is worth being blunt about the limit. I are not claiming that causal influence literally \emph{is} curvature, and Proposition~\ref{prop:nonid} says something closer to the opposite of that claim: causality is provably not hiding anywhere inside the static observational geometry, however carefully that geometry is built. What the analogy is meant to convey is only a shift in where to look: from a single shape to a family of shapes and their deformation. The rest of the paper is the attempt to make that shift precise rather than leave it as a metaphor.

\subsection{Scope}
The admissible intervention set at $x$ (the directions in which the system can, in principle, be perturbed by a given design or policy; Section~\ref{sec:response} makes this precise), and the interventional response field built on it, must be supplied from outside the model: by design, institutional knowledge, or a structural assumption. Nothing in this framework manufactures causal content from $p(x)$ and a choice of manipulable coordinates. I do not claim new asymptotic results for meta-learners or double machine learning; Section~\ref{sec:metalearners} restates known results and is explicit that "score" is used in two unrelated technical senses across the two literatures.

\subsection{Relation to Pearl's Ladder of Causation}
Pearl's Ladder of Causation separates causal reasoning into three levels: association, intervention, and counterfactual reasoning, and the objects of this paper sort cleanly onto the first two. The observational geometry $(M,g,\psi)$ belongs to the associational rung: it is determined entirely by $p(x)$ and answers questions about what tends to occur together, not about what would happen under manipulation. The interventional score $\psi_{-k}^{(k,\xi)}$, the sensitivities $S_{k\to j}$, and the response fields $S(v)$ belong to the interventional rung, since all three are indexed by an explicit $\operatorname{do}(X_k=\xi)$ and cannot generally be recovered from $p(x)$ alone. Proposition~\ref{prop:nonid} can be read as a geometric statement of exactly the gap between these two rungs: no transformation of the observational distribution, however sophisticated, substitutes for the intervention operator.

\begin{table}[ht]
\centering
\begin{tabular}{lll}
\hline
\textbf{Pearl rung} & \textbf{Object in this paper} & \textbf{Interpretation} \\
\hline
Association
& $(M,g,\psi)$
& Geometry determined by $p(x)$ \\

Intervention
& $p_{-k}^{(k,\xi)},\; S_{k\to j},\; S(v)$
& Response under $\operatorname{do}(X_k=\xi)$ \\

Counterfactual
& Not developed
& Requires unit-level cross-world structure \\
\hline
\end{tabular}
\caption{Relation between the proposed framework and Pearl's Ladder of Causation.}
\label{tab:ladder}
\end{table}

The framework does not reach the third rung, and I do not try to stretch it there. Counterfactual reasoning requires coupling factual and hypothetical outcomes at the level of a single unit, using structure such as shared exogenous disturbances in a structural causal model, and that goes beyond a family of interventional marginals indexed only by $\xi$. A geometry of that third rung is left for future work.

\subsection{Contributions}
\begin{enumerate}
\item An interventional score defined on the non-intervened coordinates, avoiding the singular-density problem of a naive $d$-dimensional gradient (Definition~\ref{def:intscore}).
\item A two-level definition of causal influence: a distributional definition (Definition~\ref{def:causal-flow}) and a matching marginal-score diagnostic (Definition~\ref{def:causal-detect}) that is sufficient for it by construction.
\item A fully worked bivariate Gaussian example, presented immediately after the core definitions rather than late in the paper, so every later section can build on numbers the reader has already seen (Section~\ref{sec:example}).
\item A proof that an admissible intervention set cannot, by itself, convert an observational score into a causal one (Proposition~\ref{prop:projwrong}), and an interventional response field that requires structural input in its place (Definition~\ref{def:response}).
\item A causal metric built as Fisher information on a family of interventions sharing a common target, so the family lives on a single, well-defined coordinate space (Definition~\ref{def:causal-metric}).
\item A corrected geometric dictionary for randomized trials (stated for both continuous and discrete treatment), instrumental variables, and conditional-independence designs (Section~\ref{sec:designs}), and an explicit terminological disambiguation for the meta-learner and double-machine-learning literature (Section~\ref{sec:metalearners}).
\end{enumerate}

\section{Related Literature}

\textbf{Information geometry and causal manifolds.} \citet{amari2016} develops the Fisher information geometry of parametric families. \citet{surasinghe2020} construct manifolds from time-delayed observations and relate geometric distances to transfer entropy and Granger causality; this is an observational and predictive program in spirit. \citet{dominguez2023} show that structural causal models with $k$ parents induce data on $k$-dimensional submanifolds, characterizing the \emph{support} of an SCM. We differ by studying how that support and its score deform under an explicit intervention operator.

\textbf{Score-based generative models.} \citet{song2021} estimate $\nabla_x\log p(x)$ directly from data via score matching and diffusion. This supplies estimates of the observational score $\psi$, not of the interventional response field of Section~\ref{sec:response}; Section~\ref{sec:scoregen} is explicit about that gap.

\textbf{Instrumental variables, meta-learners, and double machine learning.} \citet{angrist1996} give the canonical LATE identification result under monotonicity; \citet{heckman2005} develop the continuous-instrument marginal treatment effect framework, invoked in Remark~\ref{rem:iv}. \citet{kunzel2019} introduce the X-learner; \citet{chernozhukov2018} develop double machine learning via Neyman-orthogonal moments; \citet{nie2021} develop the R-learner. Section~\ref{sec:metalearners} is explicit that the "score" in Neyman-orthogonal estimation is not the density score used elsewhere in this paper.

\section{Observational Geometry and Its Limits}
\label{sec:obs-geometry}

Let $X=(X_1,\dots,X_d)$ have density $p$ on $\mathcal X\subseteq\mathbb R^d$, smooth enough for the derivatives below to exist. let $M\subseteq\mathbb R^d$ denote the smooth support of $p$, equipped with a background Riemannian metric $g$ chosen independently of $p$; typically this is just the ambient Euclidean metric in the given coordinates. We are deliberately modest about what $M$ is: it is the support of $p$ carrying a metric we supply, not a Riemannian structure that $p$ itself induces. Define the \textbf{score field}
\[
\psi(x)=\nabla_x\log p(x), \qquad x\in M,\ p(x)>0.
\]
Intuitively, $\psi(x)$ points in the direction a particle sitting at $x$ would need to move to climb the log-density landscape fastest: where $p$ is large and flat, $\psi$ is small, and where $p$ rises or falls steeply, $\psi$ is large. The triple $(M,g,\psi)$, the \textbf{observational geometry}, is a reparameterization of $p$: support, plus the local direction of increasing log-density, and nothing more than that.

\begin{proposition}[Non-identification from observational geometry]
\label{prop:nonid}
If two structural models generate the same density $p(x)$ on $M$, they induce the same observational geometry $(M,g,\psi)$; consequently any functional of $(M,g,\psi)$ that depends on the models only through $p$ cannot distinguish between them.
\end{proposition}
\begin{proof}
$M$ and $\psi(x)=\nabla_x\log p(x)$ depend on the models only through $p$, so if $p$ agrees under both models, $M$ and $\psi$ agree pointwise, and any functional built from them agrees. Section~\ref{sec:example} exhibits two structural models with different causal content and identical $p$.
\end{proof}

\section{Interventional Geometry}
\label{sec:interventional}

Most of what follows turns on getting one object right, so it is worth being unhurried about it here rather than fixing it later by patching.

\subsection{Why the interventional score must live on fewer coordinates}
For coordinate $k$ and level $\xi$, the intervention $\operatorname{do}(X_k=\xi)$ replaces the structural equation for $X_k$ with the constant $\xi$: every unit is assigned exactly the value $\xi$ on that coordinate, whatever it would otherwise have been. It helps to picture $p(x)$ as a cloud of probability mass spread over $\mathbb R^d$. Pinning one coordinate to a single value does not merely reshape that cloud; it flattens it entirely onto the $(d-1)$-dimensional slice $\{x \in \mathbb R^d : x_k = \xi\}$, much as pressing a lump of dough flat against a table leaves it spread across the tabletop rather than filling any volume above it. A distribution living on a slice like that has no density with respect to ordinary $d$-dimensional volume, since there is no thickness left in the $x_k$ direction to divide by. An expression such as $\nabla_x \log p(x \mid \operatorname{do}(X_k=\xi))$, taken naively over all $d$ coordinates, is therefore not well defined. I work instead with the density of the coordinates that remain random, which is exactly where the $d-1$ dimensions that did not get flattened away still live.

\begin{definition}[Interventional score]
\label{def:intscore}
Write $X_{-k}=(X_1,\dots,X_{k-1},X_{k+1},\dots,X_d)$, and let
\[
p^{(k,\xi)}_{-k}(x_{-k}) \;=\; p\big(X_{-k}=x_{-k} \mid \operatorname{do}(X_k=\xi)\big)
\]
be its density on $\mathbb R^{d-1}$, ordinary and non-singular wherever the intervention is well defined. The \textbf{interventional score field} is
\[
\psi^{(k,\xi)}_{-k}(x_{-k}) \;=\; \nabla_{x_{-k}}\log p^{(k,\xi)}_{-k}(x_{-k}).
\]
\end{definition}
When $d=2$, $X_{-k}$ is a single coordinate and $p^{(k,\xi)}_{-k}$ is already an ordinary one-dimensional density; Section~\ref{sec:example} is exactly this case.

Figure~\ref{fig:collapse} draws the picture for that same two-coordinate case. The left panel shows the ordinary two-dimensional density $p(x,y)$, with the vertical line marking the slice $\{X=\xi\}$ that a hard intervention on $X$ leaves behind. The right panel shows what actually survives the intervention: an ordinary one-dimensional density in $y$ alone, living along that slice. The interventional score of Definition~\ref{def:intscore} is the score of that one-dimensional density, not of anything defined on the original two-dimensional picture.

\begin{figure}[t]
\centering
\includegraphics[width=0.85\textwidth]{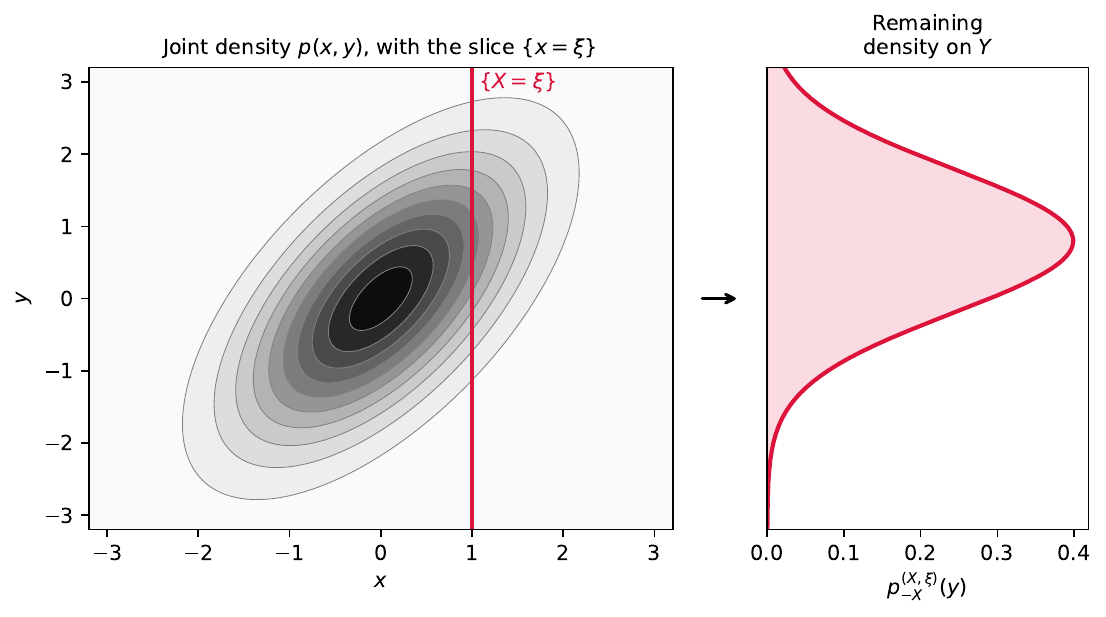}
\caption{A hard intervention $\operatorname{do}(X=\xi)$ collapses the two-dimensional density $p(x,y)$ onto the vertical slice $\{x=\xi\}$ (left). What remains is an ordinary one-dimensional density along that slice (right); this is the object Definition~\ref{def:intscore} takes the score of.}
\label{fig:collapse}
\end{figure}

\subsection{Causal influence: distributional definition and marginal diagnostic}

\begin{definition}[Causal influence]
\label{def:causal-flow}
$X_k$ has a causal influence on $X_j$ ($j\neq k$), written $X_k \rightsquigarrow X_j$, if the marginal law of $X_j$ under $\operatorname{do}(X_k=\xi)$ varies with $\xi$: there exist $\xi_1\neq\xi_2$ such that
\[
p\big(X_j \mid \operatorname{do}(X_k=\xi_1)\big) \;\neq\; p\big(X_j\mid \operatorname{do}(X_k=\xi_2)\big).
\]
\end{definition}

The diagnostic used to detect this must be built from the same marginal object, not from the joint score of all non-intervened coordinates. A nonzero derivative of $\partial_{x_j}\log p^{(k,\xi)}_{-k}(x_{-k})$, the joint score in the $x_j$ direction, can arise purely from a change in the \emph{dependence} between $X_j$ and some other non-intervened coordinate $X_\ell$, with the marginal law of $X_j$ itself left unchanged. For instance, an intervention could tighten or loosen how closely $X_j$ tracks $X_\ell$ without shifting $X_j$'s own typical behavior at all, and the joint score would register exactly that tightening, wrongly making it look as though $X_k$ had reached all the way to $X_j$. A change of that kind would not witness $X_k\rightsquigarrow X_j$ in the sense of Definition~\ref{def:causal-flow}. I therefore define the diagnostic directly on the marginal.

\begin{definition}[Score-detectable causal influence]
\label{def:causal-detect}
Let $p_j^{(k,\xi)}(x_j) = p(X_j=x_j\mid\operatorname{do}(X_k=\xi))$ be the marginal interventional density of $X_j$ alone (obtained from $p^{(k,\xi)}_{-k}$ by further marginalizing over the remaining coordinates), and define
\[
S_{k\to j}(\xi,x_j) \;=\; \frac{\partial}{\partial\xi}\,\partial_{x_j}\log p_j^{(k,\xi)}(x_j).
\]
If $S_{k\to j}(\xi,x_j)\neq 0$ on an open set of $(\xi,x_j)$, then $p(X_j\mid\operatorname{do}(X_k=\xi))$ varies with $\xi$ on that set, so $X_k\rightsquigarrow X_j$ in the sense of Definition~\ref{def:causal-flow}. The converse need not hold in general: a distributional change need not always show up as a score change, for instance if an intervention shifts or reshapes the support of $X_j$ without changing $\partial_{x_j}\log p_j^{(k,\xi)}$ at points common to both supports. There is a second, more basic reason the score is a coarser object than the density: two densities with the same score agree only up to a multiplicative constant on each connected component of a common support, and it is normalization (each $p_j^{(k,\xi)}$ must integrate to one) that pins that constant down. Proposition~\ref{prop:necessity} below shows that once the support is fixed, connected, and the density stays strictly positive, this is the \emph{only} slack in the diagnostic, and the converse we just declined to claim in general in fact holds.
\end{definition}
Because $S_{k\to j}$ is now built from the same marginal $p_j^{(k,\xi)}$ that Definition~\ref{def:causal-flow} refers to, the implication above is immediate and does not depend on the behavior of coordinates other than $X_j$. All later uses of "interventional score sensitivity" refer to this marginal $S_{k\to j}$.

\begin{proposition}[Score equivalence under a fixed support]
\label{prop:necessity}
Suppose $p_j^{(k,\xi)}$ is strictly positive and smooth in $(\xi,x_j)$ on a fixed connected support for every $\xi$ in an interval $I$. If
\[
S_{k\to j}(\xi,x_j) = 0 \qquad \text{for all } (\xi,x_j) \in I\times\text{supp}(X_j),
\]
then $p_j^{(k,\xi)}(x_j)$ does not depend on $\xi$, and hence $X_k$ does not causally influence $X_j$ in the sense of Definition~\ref{def:causal-flow} over $I$.
\end{proposition}
\begin{proof}
Fix $\xi$. The hypothesis says $\partial_{x_j}\big[\partial_\xi \log p_j^{(k,\xi)}(x_j)\big]=0$ for every $x_j$ in the support, which is connected, so $\partial_\xi \log p_j^{(k,\xi)}(x_j) = c(\xi)$ for some function $c$ not depending on $x_j$. Integrating in $\xi$ from a reference level $\xi_0\in I$,
\[
\log p_j^{(k,\xi)}(x_j) - \log p_j^{(k,\xi_0)}(x_j) = \int_{\xi_0}^{\xi} c(u)\,du =: C(\xi),
\]
again not depending on $x_j$, so $p_j^{(k,\xi)}(x_j) = e^{C(\xi)}\, p_j^{(k,\xi_0)}(x_j)$ for every $x_j$. Both sides integrate to $1$ over the fixed support, forcing $e^{C(\xi)}=1$; hence $p_j^{(k,\xi)}=p_j^{(k,\xi_0)}$ for every $\xi\in I$, i.e.\ the marginal does not vary with $\xi$.
\end{proof}
Combined with the sufficiency direction already noted in Definition~\ref{def:causal-detect}, Proposition~\ref{prop:necessity} says that on a fixed, connected, strictly positive support, $S_{k\to j}\equiv 0$ throughout $I\times\text{supp}(X_j)$ if and only if $X_k$ has no causal influence on $X_j$ over $I$: the score diagnostic is exactly as good as the distributional definition once the support itself is not moving. The gap between the two definitions, in other words, is entirely a support phenomenon. In practice this means the diagnostic can be trusted whenever an intervention is not plausibly pushing $X_j$ into values it could never have taken before; the gap is only worth worrying about when a policy might introduce genuinely new values of $X_j$ that fall outside anything observed under the status quo.

\section{A Worked Example}
\label{sec:example}

The definitions above are abstract enough that it helps to see them worked out with actual numbers before going further. The pair of models below does exactly what Proposition~\ref{prop:nonid} says is possible: two structural stories that agree on absolutely everything observational, yet disagree on the one thing this paper cares about, namely how $Y$ responds when $X$ is manipulated.

Let $X\sim N(0,1)$ and, under \textbf{Model A},
\[
Y=\beta X+\varepsilon_Y, \qquad \varepsilon_Y\sim N(0,\sigma^2),\ \varepsilon_Y\perp X,
\]
with $\beta=0.8$, $\sigma^2=1$. The joint law is bivariate normal with covariance
\[
\Sigma=\begin{pmatrix}1 & 0.8\\ 0.8 & 1.64\end{pmatrix}, \qquad \operatorname{Var}(Y)=\beta^2+\sigma^2=1.64.
\]
Because $(X,Y)$ is jointly Gaussian, the same $\Sigma$ is generated by the reverse linear model. Under \textbf{Model B},
\[
Y=\tilde\varepsilon_Y,\quad \tilde\varepsilon_Y\sim N(0,1.64), \qquad X=\gamma Y+\tilde\varepsilon_X,\quad \tilde\varepsilon_X\sim N(0,\tau^2),\ \tilde\varepsilon_X\perp Y,
\]
with $\gamma=\operatorname{Cov}(X,Y)/\operatorname{Var}(Y)=0.8/1.64\approx 0.488$ and $\tau^2=\operatorname{Var}(X)-\gamma^2\operatorname{Var}(Y)=1-0.8^2/1.64\approx 0.610$. A direct computation confirms Model B also induces $\Sigma$: a bivariate Gaussian admits a valid linear-Gaussian SCM representation in either causal direction. Models A and B are observationally equivalent, hence by Proposition~\ref{prop:nonid} induce the same $(M,g,\psi)$.

Here $d=2$, so intervening on $X$ leaves exactly one non-intervened coordinate, $Y$; the joint interventional density $p^{(X,\xi)}_{-X}(y)$ of Definition~\ref{def:intscore} and the marginal interventional density $p_Y^{(X,\xi)}(y)$ of Definition~\ref{def:causal-detect} coincide in this case, since there is nothing else to marginalize over.

Under Model A, $Y=\beta X+\varepsilon_Y$ is untouched by the intervention on $X$'s own equation, so
\[
p_Y^{(X,\xi)}(y) = N(\beta\xi,\ \sigma^2), \qquad \partial_y\log p_Y^{(X,\xi)}(y) = -\frac{y-\beta\xi}{\sigma^2}, \qquad S_{X\to Y}(\xi,y) = \frac{\beta}{\sigma^2} = 0.8 \neq 0.
\]
Under Model B, $Y=\tilde\varepsilon_Y$ does not involve $X$, so intervening on $X$ leaves $Y$'s marginal untouched:
\[
p_Y^{(X,\xi)}(y) = N(0,\ 1.64) \ \text{ for every }\xi, \qquad S_{X\to Y}(\xi,y) = 0.
\]

Figure~\ref{fig:scorefield} makes both halves of this concrete. Panel (a) plots the shared observational score field $\psi(x,y)$, the single vector field both models agree on since they generate the identical joint density. Panel (b) plots $E[Y\mid \operatorname{do}(X=\xi)]$ against the intervention level $\xi$ under each model: a line of slope $0.8$ under Model A, and a flat line at zero under Model B, exactly the two numbers computed above.

\begin{figure}[t]
\centering
\includegraphics[width=0.92\textwidth]{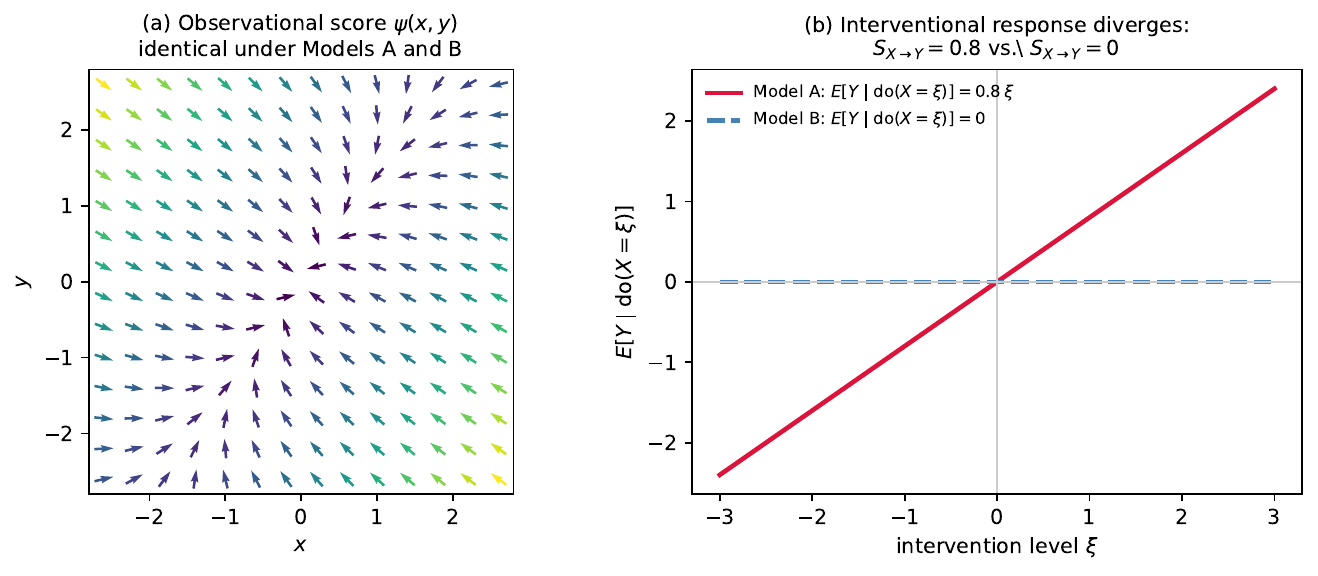}
\caption{(a) The observational score field $\psi(x,y)$ is identical under Models A and B; it is a function of the shared joint density alone, and by Proposition~\ref{prop:nonid} carries no information about which model generated it. (b) The interventional response $E[Y\mid \operatorname{do}(X=\xi)]$ diverges sharply between the two models, with slope $S_{X\to Y}=0.8$ under Model A and $S_{X\to Y}=0$ under Model B. The two panels together are the entire argument of the paper in miniature: identical statics, different dynamics.}
\label{fig:scorefield}
\end{figure}

Models A and B share the identical observational score but disagree exactly on $S_{X\to Y}$: $0.8$ under A, $0$ under B. Model A is score-detectably consistent with $X\rightsquigarrow Y$; Model B is not. This same pair reappears in Section~\ref{sec:response}, where it is used to show that a common admissible set for both models still yields a common projected score, even though the true responses differ.

\begin{remark}
Nothing in $(M,g,\psi)$ told us which model was "correct." We had to be given the structural equations to compute $S_{X\to Y}$ at all: the geometry organizes the computation, but it does not substitute for the identifying assumption.
\end{remark}

\section{Admissible Directions and the Interventional Response Field}
\label{sec:response}

Not all coordinates are equally manipulable. I record the feasible directions at $x\in M$ as a set
\[
A(x) \subseteq T_xM,
\]
representing directions in which the system can, in principle, be perturbed by design or policy. I deliberately call this an \textbf{admissible intervention set} rather than a cone: treating $A(x)$ as closed and convex silently assumes that mixtures and positive rescalings of feasible interventions are themselves feasible. That is a reasonable idealization for continuously divisible interventions (e.g., a tax rate or a dosage that can be set to any level in a range), but it need not hold for discrete or binary interventions (e.g., "treated" versus "untreated," or a fixed menu of policy options), where $A(x)$ may be a finite set with no natural convex or scalar structure. I use the convex-cone case only where it is the natural description of continuously divisible interventions; otherwise $A(x)$ should be read as an unstructured admissible set. In other words, $A(x)$ is best thought of as a menu of moves the analyst is willing to entertain at $x$, not as a piece of geometric structure that comes for free once $M$ and $g$ are fixed.

\begin{proposition}[The projected score does not identify causal response]
\label{prop:projwrong}
There exist two structural models that (i) generate the same observational density $p$, hence the same $\psi$; (ii) admit the same admissible set $A(x)$; and (iii) have genuinely different causal responses to intervention. Consequently any projection of $\psi$ onto $A(x)$ is identical under both models even though their causal content differs, and cannot be the causal field.
\end{proposition}
\begin{proof}
Take Models A and B of Section~\ref{sec:example}. They share $p(x,y)$ (hence $\psi$) by construction. Impose on both the same admissible set, namely that only the $X$-coordinate is manipulable: $A(x,y) = \operatorname{span}\{\partial_x\}\cap T_{(x,y)}M$. Both hypotheses then hold by construction, so $\Pi_{A(x)}\psi(x,y)$ is the identical vector under both models: it is, after all, a function of $\psi$ and $A(x)$ alone, and both inputs are identical by (i) and (ii). Yet Section~\ref{sec:example} shows the true response to $\operatorname{do}(X=\xi)$ differs across the models: $S_{X\to Y}=0.8$ under Model A and $=0$ under Model B. Since $\Pi_{A(x)}\psi$ takes the same value regardless of which model is true, while the causal response does not, no functional depending only on $\psi$ and $A(x)$, including the projected score, can represent the causal response in general.
\end{proof}

The admissible set tells us which directions are feasible to perturb; it does not tell us how the distribution responds. Knowing that a lever can be pulled says nothing about what happens when it is pulled, and that is the whole content of this distinction. The response is a separate primitive.

\begin{definition}[Interventional response field]
\label{def:response}
Suppose that for $v\in A(x)$ there is a structurally specified one-parameter family of interventions, with density $p^{\,\operatorname{do},\varepsilon}_v(z)$ on the coordinates $z$ left non-degenerate by the intervention, satisfying $p^{\,\operatorname{do},0}_v = p$ (or the relevant marginal or reduced density, per Section~\ref{sec:interventional}) and differentiable in $\varepsilon$ at $\varepsilon=0$. The \textbf{interventional response field along $v$}, evaluated at $x$, is the function of $z$
\[
S(v)(z) \;=\; \frac{\partial}{\partial\varepsilon}\Big[\nabla_z \log p^{\,\operatorname{do},\varepsilon}_v(z)\Big]_{\varepsilon=0}, \qquad\text{and we write } S_x(v) := S(v)(z)\big|_{z=x}
\]
for its value at the point of interest, when $x$ remains among the non-degenerate coordinates.
\end{definition}
For intuition, $v$ might stand for "increase the dosage." The family $p^{\,\operatorname{do},\varepsilon}_v$ is then the sequence of outcome distributions produced by dosages of increasing intensity $\varepsilon$, and $S(v)$ measures how fast that sequence moves the outcome's score right at the status quo, $\varepsilon=0$. Causal information lives in the map $v\mapsto S(v)(\cdot)$, which requires the family $\{p^{\,\operatorname{do},\varepsilon}_v\}$ as an input: structural or experimental knowledge of how the system actually responds to manipulation, not merely which directions are available to move in. It is not recoverable from $\psi$ and $A(x)$ alone (Proposition~\ref{prop:projwrong}).

\subsection{Location interventions as translations of the score field}
\label{subsec:location-shifts}

The response field of Definition~\ref{def:response} is deliberately abstract, so it is worth seeing it worked out in a case simple enough to hold in one's head. Many interventions of practical interest, such as a subsidy that adds a fixed amount to everyone's income or a dosage shift applied uniformly across patients, act by sliding a distribution along the number line without reshaping it. Suppose an intervention indexed by \(\xi\) changes the marginal distribution of a response variable \(Y\) through exactly such a location shift:
\begin{equation}
p_\xi(y)
=
p_0\bigl(y-m(\xi)\bigr),
\label{eq:location-family}
\end{equation}
where \(m(\xi)\) is differentiable and gives the size of the shift at intervention level $\xi$. Let
\begin{equation}
\ell_0(u)
=
\log p_0(u).
\end{equation}
The score of the interventional distribution is then
\begin{align}
\psi_\xi(y)
&=
\partial_y \log p_\xi(y) \\
&=
\partial_y \ell_0\bigl(y-m(\xi)\bigr) \\
&=
\ell_0'\bigl(y-m(\xi)\bigr).
\label{eq:location-score}
\end{align}
Differentiating with respect to the intervention level gives
\begin{align}
S(\xi,y)
&=
\partial_\xi \psi_\xi(y) \\
&=
\partial_\xi
\ell_0'\bigl(y-m(\xi)\bigr) \\
&=
-m'(\xi)\,
\ell_0''\bigl(y-m(\xi)\bigr).
\label{eq:location-response}
\end{align}

Equation~\eqref{eq:location-response} separates the interventional response into two components with a clean reading. The term \(m'(\xi)\) measures the rate at which the intervention translates the distribution (how fast the subsidy grows, say), while \(\ell_0''\) describes the local shape of the log-density, i.e.\ how sharply peaked the untouched distribution $p_0$ is at the point being shifted through. A location intervention moves the probability landscape without changing its shape, so the score's response to it factors cleanly into "how fast are we moving it" times "how steep is the landscape we are moving it through." Nothing about dependence structure or higher moments enters at all, which is exactly why this special case is simple enough to compute by hand.

For a Gaussian location family,
\begin{equation}
p_\xi(y)
=
\mathcal{N}\bigl(m(\xi),\sigma^2\bigr),
\end{equation}
we have
\begin{equation}
\ell_0''(u)
=
-\frac{1}{\sigma^2}.
\end{equation}
Therefore,
\begin{equation}
S(\xi,y)
=
\frac{m'(\xi)}{\sigma^2}.
\label{eq:gaussian-location-response}
\end{equation}
In the bivariate Gaussian example of Section~\ref{sec:example}, \(m(\xi)=\beta \xi\), so
\begin{equation}
S_{X\rightarrow Y}(\xi,y)
=
\frac{\beta}{\sigma^2}.
\end{equation}
This is precisely the interventional score sensitivity computed by hand in Section~\ref{sec:example}, which is reassuring: it means that earlier calculation was not a coincidence specific to that example but an instance of this general location-family formula.

\section{A Causal Metric}
\label{sec:metric}

A metric earns its keep only if it answers a simple question in a single number: how much does a small twist of the intervention knob actually move the distribution? Fisher information is the natural way to measure that, but it needs a well-defined family of distributions to measure it \emph{within}. A causal metric built directly from cross-target derivatives runs into a further issue beyond the ones already corrected: an intervention on $X_k$ leaves a density on $X_{-k}$, while an intervention on a different coordinate $X_{k'}$ leaves a density on $X_{-k'}$, and these are not the same coordinate space. Fisher information requires a single dominated family on a common measurable space, so comparing sensitivities across different intervention targets is not simply a matter of writing a joint expectation. I therefore build the metric one target at a time.

\begin{definition}[Causal metric for a fixed intervention target]
\label{def:causal-metric}
Fix a target coordinate $k$. Suppose interventions on $X_k$ are indexed by $\eta=(\eta_1,\dots,\eta_m)\in H$ (the leading case is $m=1$ with $\eta=\xi$, but $\eta$ may parameterize a richer family of intervention protocols on $X_k$), generating densities $\{p^{(k)}_\eta\}_{\eta\in H}$ on the common space of non-intervened coordinates $X_{-k}\in\mathbb R^{d-1}$, dominated by a common measure and smooth in $\eta$ with finite Fisher information. Define
\[
G^{(k)}_{ab}(\eta) \;=\; E_{p^{(k)}_\eta}\!\left[\partial_{\eta_a}\log p^{(k)}_\eta(X_{-k})\,\partial_{\eta_b}\log p^{(k)}_\eta(X_{-k})\right].
\]
\end{definition}
A large $G^{(k)}(\eta)$ means that even a small change in the intervention level $\eta$ produces a density that is easy to tell apart, statistically, from the one just before it: the intervention is, in a precise sense, informative there. A small $G^{(k)}(\eta)$ means the opposite: nearby intervention levels are hard to distinguish from data, so an experiment run at that $\eta$ will need a much larger sample before anything shows up.

\paragraph{Fisher geometry of a location-intervention family.}

The location family of Equation~\eqref{eq:location-family} is transparent enough to compute the fixed-target causal metric by hand, which is worth doing once before trusting Definition~\ref{def:causal-metric} in messier settings. Differentiating the log-density with respect to the intervention parameter gives
\begin{align}
\partial_\xi \log p_\xi(y)
&=
\partial_\xi
\log p_0\bigl(y-m(\xi)\bigr) \\
&=
-m'(\xi)\,
\partial_y
\log p_0\bigl(y-m(\xi)\bigr).
\end{align}
The corresponding Fisher information metric is
\begin{align}
G(\xi)
&=
\mathbb{E}_{p_\xi}
\left[
\left(
\partial_\xi \log p_\xi(Y)
\right)^2
\right] \\
&=
\bigl(m'(\xi)\bigr)^2
\mathbb{E}_{p_\xi}
\left[
\left(
\partial_y
\log p_0\bigl(Y-m(\xi)\bigr)
\right)^2
\right].
\end{align}
After the change of variable
\begin{equation}
U=Y-m(\xi),
\end{equation}
we obtain
\begin{equation}
G(\xi)
=
\bigl(m'(\xi)\bigr)^2
I_{\mathrm{loc}}(p_0),
\label{eq:location-fisher-metric}
\end{equation}
where
\begin{equation}
I_{\mathrm{loc}}(p_0)
=
\mathbb{E}_{p_0}
\left[
\left(
\partial_u \log p_0(U)
\right)^2
\right]
\end{equation}
is the Fisher information for the location parameter.

If the intervention is a direct translation, \(m(\xi)=\xi\), then
\begin{equation}
G(\xi)
=
I_{\mathrm{loc}}(p_0),
\end{equation}
which is constant in \(\xi\). For a Gaussian distribution with variance \(\sigma^2\),
\begin{equation}
I_{\mathrm{loc}}(p_0)
=
\frac{1}{\sigma^2},
\end{equation}
so the line element on the intervention-parameter space is
\begin{equation}
ds^2
=
\frac{1}{\sigma^2}\,d\xi^2.
\end{equation}

More generally, for any regular one-dimensional intervention family, define the arc-length coordinate
\begin{equation}
s(\xi)
=
\int_{\xi_0}^{\xi}
\sqrt{G(u)}\,du.
\end{equation}
In the coordinate \(s\), the line element becomes
\begin{equation}
ds^2.
\end{equation}
Thus every regular one-dimensional intervention family is locally Euclidean after arc-length reparameterization. This flatness is worth not over-reading: any smooth, strictly positive metric on a one-dimensional space becomes the standard Euclidean line element after an arc-length change of coordinates. That is a generic fact of one-dimensional Riemannian geometry, true of every regular curve whatsoever, and it is not a special discovery about interventions or about causal structure. What the calculation actually contributes is $G(\xi)$ itself, computed \emph{before} that reparameterization flattens it away: the decomposition into the intervention's own "speed" $m'(\xi)$ and the ambient sharpness of the log-density, $I_{\mathrm{loc}}(p_0)$, is the substantive content, and it is a statement about the geometry of the intervention-parameter space, not about the curvature of the observational manifold or of the complete causal system.

This is the ordinary Fisher information metric of the family $\{p^{(k)}_\eta\}$, well defined under the stated regularity conditions, and it is defined separately for each target $k$, so every member of the family it describes lives on the same space $\mathbb R^{d-1}$. Comparing sensitivity across two different targets $k\neq k'$ requires either restricting attention to the coordinates common to both $X_{-k}$ and $X_{-k'}$ (i.e., excluding both $k$ and $k'$), or moving to soft or stochastic interventions that preserve the full $d$-dimensional support and hence a common space throughout; I flag the latter as the more natural fix but do not develop it here.

\section{Classical Designs, Corrected}
\label{sec:designs}

\subsection{Randomized controlled trials}
Let $(T,Y,Z)$ be treatment, outcome, and covariates, with $T$ and $Z$ continuously distributed with positive smooth density. Define $\kappa_{T,Z}(t,z)=\partial_t\partial_z\log p(t,z)$.

\begin{proposition}
\label{prop:rct}
For continuously distributed $T,Z$ with positive smooth density on a connected, rectangular support: $T\perp Z$ if and only if $\kappa_{T,Z}(t,z)=0$ for all $(t,z)$.
\end{proposition}
\begin{proof}
If $T\perp Z$, $p(t,z)=p(t)p(z)$, so $\log p(t,z)=a(t)+b(z)$ and the mixed partial vanishes. Conversely, on a connected rectangular domain, $\kappa_{T,Z}\equiv 0$ forces $\log p(t,z)=a(t)+b(z)$ (mixed-partials test), giving $T\perp Z$.
\end{proof}

The quantity $\kappa_{T,Z}$ is not a separate object worth tracking on its own; it is simply the derivative, along the treatment coordinate, of the covariate component of the joint score field
\begin{equation}
\psi_{T,Z}(t,z)
=
\nabla_{(t,z)}\log p(t,z)
=
\begin{pmatrix}
\partial_t\log p(t,z)\\[4pt]
\nabla_z\log p(t,z)
\end{pmatrix}, \qquad \kappa_{T,Z}(t,z) = \partial_t \nabla_z \log p(t,z).
\end{equation}
Under random assignment, $T\perp Z$ gives $p(t,z)=p_T(t)p_Z(z)$, so $\log p(t,z)=\log p_T(t)+\log p_Z(z)$, and the score above separates block-wise into a piece depending only on $t$ and a piece depending only on $z$:
\begin{equation}
\psi_{T,Z}(t,z)
=
\begin{pmatrix}
\partial_t \log p_T(t) \\
\nabla_z \log p_Z(z)
\end{pmatrix}.
\end{equation}
Moving along the treatment coordinate therefore leaves the covariate component of the score untouched, and moving along the covariate coordinates leaves the treatment component untouched, which is exactly $\kappa_{T,Z}\equiv 0$ again, now visible directly in the block structure of $\psi_{T,Z}$ rather than just asserted. Geometrically, the treated and control arms begin from the same covariate score field, the same pretreatment probability landscape; whatever happens to the outcome afterward cannot be traced back to a systematic difference in how $Z$ was distributed at the moment of assignment.

\begin{remark}[Discrete or binary treatment]
Most randomized trials assign a discrete or binary $T$, for which $\partial_t$ is not defined in the ordinary sense and Proposition~\ref{prop:rct} does not directly apply. The corresponding statement is: $T\perp Z$ if and only if $p(z\mid T=t)$ does not depend on $t$, equivalently
\[
\nabla_z\log p(z\mid T=1) \;=\; \nabla_z\log p(z\mid T=0)
\]
in the binary case, and analogously across all levels of $T$ in the general discrete case. This is the natural discrete analogue of $\kappa_{T,Z}\equiv 0$ and is the version relevant to most actual randomized trials: the treated and control arms start out statistically indistinguishable on covariates, up to finite-sample imbalance that chance alone can still produce.
\end{remark}

This block-separation is easy to over-read, so it is worth being blunt about what it does not give us. Random assignment delivers $T\perp Z$; it does not, and should not, deliver $T\perp Y$: the entire point of running the trial is to let $T$ affect $Y$, so $\partial_t\partial_y\log p(t,y,z)$ is generally, and intentionally, nonzero. Nor should this be confused with vanishing Riemannian curvature of the full manifold. What randomization buys geometrically is exactly the block-separation above, applied to the assignment mechanism only, and it says nothing about the treatment-outcome relationship the trial exists to measure; in particular, it does not flatten that relationship.

Randomization and intervention therefore play distinct geometric roles, and it is worth closing this subsection by naming the difference plainly. Randomization flattens the coupling in the assignment mechanism, $\partial_t\nabla_z\log p(t,z)=0$: it certifies that the two arms started from a level field. An actual intervention on $X_k$, by contrast, generates a trajectory of response distributions $\xi\mapsto p_j^{(k,\xi)}(x_j)$, whose local deformation is exactly $S_{k\to j}$ from Definition~\ref{def:causal-detect}. A randomized trial tells you the starting line was fair; the response field is what tells you how far the race actually moved.

\subsection{Conditional independence designs}
Under $Y(t)\perp T\mid Z$ for $t\in\{0,1\}$, standard potential-outcomes arguments \citep{rosenbaum1983,imbens2015} show that $E[Y\mid T=t,Z=z]$ identifies $E[Y(t)\mid Z=z]$. I do not have a geometric restatement of this that is both non-circular and adds content beyond the standard statement, and I do not offer one.

\subsection{Instrumental variables}
\label{sec:iv}
Let $Z$ be a candidate instrument for treatment $T$ on outcome $Y$, with $v_Z$ a direction of variation in $Z$. Since $T$ and $Y$ are random variables rather than deterministic functions of $Z$, the relevant well-defined objects are the conditional mean functions $m_Y(z)=E[Y\mid Z=z]$, $m_T(z)=E[T\mid Z=z]$, which are ordinary functions on the $Z$-manifold with genuine gradients.

\begin{remark}[IV identity]
\label{rem:iv}
Under exclusion and relevance ($g(v_Z,\nabla m_T(z))\neq 0$), the local instrumental-variables estimand along $v_Z$ is
\[
\beta_{\mathrm{IV}}(z) \;=\; \frac{D_{v_Z}m_Y(z)}{D_{v_Z}m_T(z)} \;=\; \frac{g(v_Z,\nabla m_Y(z))}{g(v_Z,\nabla m_T(z))}.
\]
For a binary instrument with monotonicity, the difference-form version of this ratio equals the local average treatment effect (LATE) among compliers \citep{angrist1996}. For a continuous instrument, the derivative version above is a local IV estimand in its own right; it is not automatically a marginal treatment effect. Under the latent-index treatment-selection structure and further assumptions used in that literature \citep{heckman2005}, it may be connected to a weighted or localized marginal treatment effect, but that connection is not immediate from the directional-derivative form alone.
\end{remark}

\subsection{Graphical interventions}
In a DAG-based structural causal model, Pearl's $\operatorname{do}$-operator replaces the structural equation of the intervened node with a constant and deletes its incoming edges. This is exactly the map from $p$ to $p^{(k,\xi)}_{-k}$ of Definition~\ref{def:intscore}: the intervened coordinate is fixed, not merely reweighted, and drops out of the density.

\begin{table}[t]
\centering
\small
\begin{tabular}{@{}p{2.6cm}p{4.6cm}p{6.8cm}@{}}
\toprule
\textbf{Design} & \textbf{Classical statement} & \textbf{Geometric restatement} \\
\midrule
RCT (continuous $T$) & $T\perp Z$ by assignment & $\kappa_{T,Z}\equiv 0$ (Prop.~\ref{prop:rct}); does not extend to $T$ vs.\ $Y$ \\
RCT (discrete $T$) & $T\perp Z$ by assignment & $\nabla_z\log p(z\mid T=t)$ constant across $t$ \\
CIA & $Y(t)\perp T\mid Z$ & standard identification of $E[Y(t)\mid Z]$; no geometric restatement offered \\
IV & Wald ratio $=$ reduced form / first stage & $D_{v_Z}m_Y(z)/D_{v_Z}m_T(z)$; equals LATE only under discrete-instrument and monotonicity conditions \\
Graphical $\operatorname{do}$ & edge deletion, structural equation replaced & $p \mapsto p^{(k,\xi)}_{-k}$, intervened coordinate removed from the density \\
\bottomrule
\end{tabular}
\caption{The geometric dictionary.}
\end{table}

\section{Score-Based Generative Models}
\label{sec:scoregen}

It is tempting to think that a sufficiently good generative model, trained on enough observational data, would eventually "learn" causal structure as a byproduct of learning the density well. Section~\ref{sec:response} rules this out directly, and it is worth restating why in this more concrete setting before moving on. Score-matching and diffusion-based generative models estimate $\nabla_x\log p(x)$ from data by fitting $s_\theta(x)$ and sampling via $dX_t = s_\theta(X_t)\,dt+\sqrt2\,dW_t$ \citep{song2021}, supplying $\hat\psi(x)$, an estimate of the \emph{observational} score. Given Proposition~\ref{prop:projwrong}, projecting $\hat\psi$ onto an admissible set does not yield a causal estimate, no matter how accurate $\hat\psi$ becomes: the family $\{p^{\,\operatorname{do},\varepsilon}_v\}$ that causal content actually lives in is simply not a fact about $p(x)$, so no amount of fitting $p(x)$ more precisely brings it into view. If interventional or quasi-experimental data are available (samples from $p^{\,\operatorname{do},\varepsilon}_v$ for $\varepsilon$ near $0$), one could in principle estimate the interventional response field $S(v)$ of Definition~\ref{def:response} by score-matching each member of the interventional family and differencing. This is a data requirement beyond observational score estimation, not a byproduct of it.

\section{A Terminological Note on "Score"}
\label{sec:metalearners}

The word "score" is overloaded across the literatures this paper touches, and it is worth being explicit about that rather than building a table that only papers over it. In this paper, "score" means $\nabla_x\log p(x)$, the gradient of a log-density. In semiparametric efficiency theory and in double machine learning, "score" (or "estimating equation") means a function $\psi(W;\theta,\eta)$ satisfying a moment condition $E[\psi(W;\theta_0,\eta_0)]=0$, chosen for Neyman orthogonality $\partial_\eta E[\psi(W;\theta_0,\eta)]|_{\eta=\eta_0}=0$ \citep{chernozhukov2018}. In classical statistics, "score" can also mean the likelihood score $\partial_\theta \log p_\theta(x)$ of a parametric family (related to our usage, but indexed by a parameter rather than by the data coordinates themselves). These are three distinct objects, related only by an accident of terminology, and conflating them is an easy way to make a false claim sound like a derivation.

With that said: S-, T-, and X-learners \citep{kunzel2019}, double machine learning \citep{chernozhukov2018}, and the R-learner \citep{nie2021} are all, under their respective identifying assumptions, estimators of the same conditional average treatment effect $E[Y(1)-Y(0)\mid X=x]$, a \emph{contrast} rather than a derivative in the binary-treatment case that motivates most of these methods. (For a continuously parameterized treatment, an analogous derivative object can be defined under additional smoothness assumptions, but the R-learner's usual target is the discrete contrast.) Their consistency results belong to the cited papers; I do not derive anything new about them here, and I would rather leave them uncatalogued than dress them in geometric language that adds nothing beyond what a plain restatement already says.

\section{Discussion: Scope and Limitations}

\begin{itemize}
\item \textbf{The admissible set and the response family are inputs, not outputs.} Every causal claim traces back to $A(x)$ and $\{p^{\,\operatorname{do},\varepsilon}_v\}$, both external to $p(x)$.
\item \textbf{$A(x)$ need not be convex or closed under scaling.} I use cone structure only when interventions are continuously divisible; for discrete or binary interventions $A(x)$ is better read as an unstructured admissible set.
\item \textbf{The causal metric is defined per intervention target.} Comparing sensitivity across two different targets requires either restricting to their common non-intervened coordinates or moving to soft interventions that preserve a common support; I do not develop either extension here.
\item \textbf{Definition~\ref{def:causal-detect} is sufficient, and, on a fixed, connected, strictly positive support, also necessary}, by Proposition~\ref{prop:necessity}. The remaining gap between the score diagnostic and the distributional definition is entirely a support phenomenon: it can only appear when an intervention moves, shrinks, or disconnects the support of $X_j$.
\item \textbf{The worked example is deliberately simple.} Higher-dimensional, non-Gaussian, or nonlinear settings raise identifiability and estimation issues (e.g., LiNGAM-type non-Gaussian identifiability, weak instruments, positivity failures at the boundary of $A(x)$) not addressed here.
\end{itemize}

\section{Conclusion}

Observational geometry cannot identify causal direction, because it is a function of $p(x)$ alone (Proposition~\ref{prop:nonid}). Building the interventional analogue requires working on the coordinates left non-degenerate by a hard intervention (Definition~\ref{def:intscore}), defining causal influence distributionally with a matching marginal-score diagnostic (Definitions~\ref{def:causal-flow} and~\ref{def:causal-detect}), and recognizing that an admissible intervention set cannot by itself convert an observational score into a causal one (Proposition~\ref{prop:projwrong}); the causal content instead lives in an interventional response field that must be supplied structurally (Definition~\ref{def:response}). A causal metric built from Fisher information on a fixed-target intervention family is well posed without the cross-target space mismatch of a naive construction (Definition~\ref{def:causal-metric}). Applied to randomized trials, instrumental variables, and conditional-independence designs, this vocabulary reproduces exactly what each design implies and nothing more: randomization decouples treatment from covariates but says nothing about the treatment-outcome relationship it is designed to reveal, and the instrumental-variables identity is a ratio of gradients of conditional mean functions rather than of the outcome and treatment variables themselves. The framework is a common notation for relating interventions, admissible directions, and score fields across structural econometrics, experimental design, and score-based generative modeling; it does not add identification power beyond what the underlying assumptions already supply. One loose end is now tied off: the score diagnostic is not merely sufficient but also necessary for causal influence once the support of the affected variable holds still (Proposition~\ref{prop:necessity}), so the diagnostic and the definition can only part ways when a support itself moves, shrinks, or splits apart under intervention. A metric that compares sensitivity across different intervention targets, rather than one target at a time, remains open, and would be the natural next thing to build. In the language of Pearl's Ladder of Causation, the paper geometrically separates association from intervention while leaving the counterfactual rung, which requires cross-world, unit-level structural information, as a natural extension.

\bibliographystyle{plainnat}

\end{document}